\documentclass[copyright,creativecommons]{eptcs}
\usepackage{\jobname}

\providecommand{\event}{GandALF 2014} 
\usepackage{breakurl}             

\title{Complexity of validity for propositional dependence logics
}
\author{Jonni Virtema
\institute{
School of Information Science\\
Japan Advanced Institute of Science and Technology\\
Nomi, Japan}
\institute{
School of Information Sciences\\
University of Tampere\\
Tampere, Finland}
\email{jonni.virtema@uta.fi}}

\def\titlerunning{Complexity of validity for propositional dependence logics}
\def\authorrunning{Jonni Virtema}

\begin{document}

\maketitle

\begin{abstract}
We study the validity problem for propositional dependence logic, modal dependence logic and extended modal dependence logic. We show that the validity problem for propositional dependence logic is $\NEXPTIME$-complete. In addition, we establish that the corresponding problem for modal dependence logic and extended modal dependence logic is $\NEXPTIME$-hard and in $\NEXPTIME^\NP$.
\end{abstract}
%
\section{Introduction}
Dependencies occur in many scientific disciplines. For example, in physics there are dependencies in experimental data,
and in social science they can occur between voting extrapolations.
For example, one might want to express whether a value of a certain physical measurement is determined by the values of some other measurements. More concretely, is it the case that in some collection of experimental data, the temperature of some object is completely determined by the solar activity and the distance between the object and the sun. 
One might also want to know whether the voting pattern of some single constituency always determines the election results.

With the aim to express such dependencies Väänänen introduced first-order \emph{dependence logic} \cite{va07} and its modal variant \emph{modal dependence logic} \cite{va08}. First-order dependence logic extends first-order logic by novel atomic formulae called \emph{dependence atoms}.
Modal dependence logic, in turn, extends modal logic with \emph{propositional dependence atoms}. A dependence atom, denoted by $\fdep{x_1,\dots,x_n,y}$, intuitively states that the value of the variable $y$ is solely determined by the values of the variables $x_1,\dots,x_{n}$. The intuitive meaning of the propositional dependence atom $\dep{p_1,\dots,p_n,q}$ is that the truth value of the proposition $q$ is functionally determined by the truth values of the propositions $p_1,\dots,p_{n}$.
One of the core ideas in these logics of dependence is the use of team semantics. Väänänen realized that dependencies do not manifest themselves in a single assignment nor in a single point. To manifest dependencies one must look at sets of assignment or collections of points. These sets of assignments or points are called teams.
Thus whereas in the standard semantics for first-order logic 
formulae are evaluated with respect to first-order models and assignments, in team semantics of dependence logic formulae are evaluated with respect to first-order models and sets of assignments. Analogously, in team semantics for modal logic formulae are evaluated with respect to Kripke models and sets of points.
For example, the formula
\[
\fdep{x_{\mathit{activity}}, x_{\mathit{dist}}, x_{\mathit{temp}}},
\]
where the values of the variables $x_\mathit{activity}$, $x_{\mathit{dist}}$, and $x_{\mathit{temp}}$ range over the magnitude of solar activity, distance to the sun, and temperature, respectively,
expresses that in some set of data the temperature is completely determined by the solar activity and the distance to the sun. Sets of data are captured by teams. Each assignment in a team corresponds to one record of data.

Team semantics was originally defined by Hodges \cite{Hodges97c} as a means to obtain compositional semantics for the independence-friendly logic of Hintikka and Sandu \cite{hisa89}. Later on Väänänen adopted team semantics as a central notion for his dependence logic.

Modal dependence logic was the first step in combining functional dependence and modal logic. The logic however lacks the ability to express temporal dependencies; there is no mechanism in modal dependence logic to express dependencies that occur between different points of the model. This is due to the restriction that only proposition symbols are allowed in the dependence atoms of modal dependence logic. To overcome this defect Ebbing et~al. \cite{EHMMVV13} introduced the
\emph{extended modal dependence logic} by extending the scope of dependence atoms to arbitrary modal formulae, i.e., dependence atoms in extended modal dependence logic are of the form $\dep{\varphi_1,\dots\varphi_n,\psi}$, where $\varphi_1,\dots,\varphi_n,\psi$ are formulae of modal logic. For example when interpreted in a temporal model, the formula
\[
\dep{\Diamond_P\, q, \Diamond_P\Diamond_P \,q, \Diamond_P\Diamond_P\Diamond_P \,q,\,q}
\]
expresses that the truth of $q$, at this moment, only depends of the truth of $q$ in the previous $3$ time steps.

It was shown in \cite{EHMMVV13} that extended modal dependence logic is strictly more expressive than modal dependence logic. Furthermore Hella et~al. \cite{HeLuSaVi14} established that exactly the properties of teams that are downward closed and closed under the so-called team $k$-bisimulation, for some finite $k$, are definable in extended modal dependence logic. The characterization of Hella et~al. truly demonstrates the naturality of extended modal dependence logic. In recent years the research around modal dependence logic has bloomed, for recent work see e.g.
\cite{EHMMVV13,EbbingL12,ebloya,galliani13,lohvo13,vollmer13,Sevenster:2009}.

Team semantics in propositional context is also closely related to the inquisitive logic of Groenendijk \cite{Groenendijk07}. In inquisitive logic the meaning of formulae is defined on sets of assignments for proposition
symbols. This connection between propositional dependence logic and inquisitive logic has already been noted in the recent Ph.D. thesis of Fan Yang \cite{fanthesis}. For resent work related to inquisitive logic, see e.g. \cite{Ciardelli11,sano11}.

In this paper we study the computational complexity of the validity problem for propositional dependence logic, modal dependence logic and extended modal dependence logic.
The study of computational complexity of the satisfiability problem and the model checking problem for logics of dependence has been very active.
For research related to fragments of first-order dependence logic and related formalisms see \cite{blassgure, jarmo, KoKuLoVi:2011, Graedel:2013, jonnithesis}. For work on variants of propositional and modal dependence logics see \cite{Sevenster:2009, EbbingL12, EHMMVV13, lohvo13, fanthesis}.
However, there is not much research done on the validity problem of these logics. We wish to mend this shortcoming. Note that since the logics of dependence are not closed under negation, the traditional connection between the satisfiability problem and the validity problem fails.
In this article we establish that the validity problem for propositional dependence logic is  $\NEXPTIME$-complete. In addition, we obtain that the corresponding problem for modal dependence logic and extended modal dependence logic is contained in $\NEXPTIME^\NP$.

The article is structured as follows. In section \ref{preli} we define the basic concepts and results relevant to this article. In section \ref{dqbf} we introduce a variant of $\QBF$, called dependency quantified Boolean formulae, for which the decision problem whether a given formula is true is $\NEXPTIME$-complete. We start Section \ref{ccpd} with compact definitions of satisfiability, validity and model checking in the context of team semantics. The rest of the section is devoted for the study of the complexity of the validity problem for propositional dependence logic. In Section \ref{ccml} we consider the validity problem of modal dependence logic and extended modal dependence logic.
%
\section{Preliminaries}\label{preli}
In this section we define the basic concepts and results relevant to this article. We assume that the reader is familiar with propositional logic $\PL$ and modal logic $\ML$.

\subsection{Propositional logics}
Let $\mathbb{Z}_+$ denote the set of positive integers, and let $\mathrm{PROP}=\{p_i\mid i\in \mathbb{Z}_+\}$ be the set of exactly all \emph{proposition symbols}. Let $D$ be a finite, possibly empty, subset of $\mathrm{PROP}$. A function $s:D\to \{0,1\}$ is called an \emph{assignment}. A set $X$ of assignments $s:D\to \{0,1\}$ is called a \emph{propositional team}. The set $D$ is the \emph{domain} of $X$. Note that the empty team $\emptyset$ does not have a unique domain; any subset of $\mathrm{PROP}$ is a domain of the empty team.

Most of the logics considered in this article are not closed under negation, thus we adopt the convention that a syntax of a logic is always defined in negation normal form, i.e., negations are allowed only in front of proposition symbols. This convention is widely used in the dependence logic community. Formula that is not in negation normal form is regarded as a shorthand for the formula obtained by pulling all the negations to the atomic level.

Let $\Phi$ be a set of proposition symbols. The syntax for propositional logic $\PL(\Phi)$ is defined as follows.
\[
\varphi \ddfn p\mid \neg p \mid (\varphi \wedge \varphi) \mid (\varphi \vee \varphi),
\]
where $p\in\Phi$. We will now give the team semantics for propositional logic. 
As we will see below, the team semantics and the ordinary semantics for propositional logic defined via assignments, in a rather strong sense, coincide.
\begin{definition}
Let $\Phi$ be a set of atomic propositions and let $X$ be a propositional team. The satisfaction relation $X\models \varphi$ is defined as follows. Note that, we always assume that the proposition symbols that occur in $\varphi$ are also in the domain of $X$.
\begin{align*}
X\models p  \quad\Leftrightarrow\quad& \forall s\in X: s(p)=1.\\
X\models \neg p \quad\Leftrightarrow\quad& \forall s\in X: s(p)=0.\\
X\models (\varphi\land\psi) \quad\Leftrightarrow\quad& X\models\varphi \text{ and } X\models\psi.\\
X\models (\varphi\lor\psi) \quad\Leftrightarrow\quad& Y\models\varphi \text{ and } 
Z\models\psi,
\text{ for some $Y,Z$ such that $Y\cup Z= X$}.
\end{align*}
\end{definition}
\begin{proposition}[\cite{Sevenster:2009}]\label{PLflat}
Let $\varphi$ be a formula of propositional logic and let $X$ be a propositional team. Then 
\[
X\models \varphi \quad\text{ iff }\quad \forall s\in X: s\models_\mathit{PL} \varphi.
\]
Here $\models_\mathit{PL}$ refers to the ordinary satisfaction relation of propositional logic defined via assignments.
\end{proposition}
The syntax of \emph{propositional dependence logic} $\PD(\Phi)$ is obtained by extending the syntax of $\PL(\Phi)$ by the grammar rule
\[
\varphi \ddfn \dep{p_1,\dots,p_n,q},
\]
where $p_1,\dots,p_n,q\in\Phi$. The intuitive meaning of the \emph{propositional dependence atom} $\dep{p_1,\dots,p_n,q}$ is that the truth value of the proposition symbol $q$ solely depends on the truth values of the proposition symbols $p_1,\dots,p_n$. The semantics for the propositional dependence atom is defined as follows:
\begin{align*}
X\models \dep{p_1,\dots,p_n,q} \quad\Leftrightarrow\quad& \forall s,t\in X: s(p_1)=t(p_1), \dots, s(p_n)=t(p_n)\\
&\text{implies that } s(q)=t(q).
\end{align*}
The next proposition is very useful. The proof is very easy and the result is stated, for example, in \cite{fanthesis}.
\begin{proposition}[Downwards closure]\label{dcprop}
Let $\varphi$ be a formula of propositional dependence logic and let $Y\subseteq X$ be propositional teams. Then 
$X\models \varphi$ implies $Y\models \varphi$.
\end{proposition}
\subsection{Modal logics}
In this article, in order to keep the notation light, we restrict our attention to mono-modal logic, i.e., to modal logic with just two modal operators ($\Diamond$ and $\Box$). However this is not really a restriction, since the definitions, results, and proofs of this article generalize, in a straightforward manner, to handle also the poly-modal case.

Let $\Phi$ be a set of atomic propositions. The set of formulae for \emph{standard mono-modal logic} $\ML(\Phi)$ is generated by the following grammar
\[
\varphi \ddfn p\mid \neg p \mid (\varphi \wedge \varphi) \mid (\varphi \vee \varphi) \mid \Diamond \varphi \mid \Box \varphi,
\]
where $p\in\Phi$. Note that, since negations are allowed only in front of proposition symbols, $\Box$ and $\Diamond$ are \emph{not} interdefinable.
The syntax of \emph{modal logic with intuitionistic disjunction} $\ML(\varovee)(\Phi)$ is obtained by extending the syntax of $\ML(\Phi)$ by the grammar rule
\[
\varphi \ddfn (\varphi\varovee\varphi).
\]
The \emph{team semantics for modal logic} is defined via \emph{Kripke models} and \emph{teams}. In the context of modal logic, teams are subsets of the domain of the model.
\begin{definition}
Let $\Phi$ be a set of atomic proposition symbols. A \emph{Kripke model} $\mathrm{K}$ over $\Phi$ is a tuple $\mathrm{K} = (W, R, V)$, where $W$ is a nonempty set of \emph{worlds}, $R\subseteq W\times W$ is a binary relation, and $V\colon \Phi \to \mathcal{P}(W)$ is a \emph{valuation}. A subset $T$ of $W$ is called a \emph{team} of $\mathrm{K}$. Furthermore, define that
\begin{align*}
R[T] &:= \{w\in W \mid vRw \text{ holds for some }v\in T\},\\
R^{-1}[T] &:= \{w\in W \mid wRv \text{ holds for some }v\in T\}.
\end{align*}
For teams $T,S\subseteq W$, we write $T[R]S$ if $S\subseteq R[T]$ and $T\subseteq R^{-1}[S]$. Thus, $T[R]S$ holds if and only if for every $w\in T$ there exists some $v\in S$ such that $wRv$, and for every $v\in S$ there exists some $w\in T$ such that $wRv$.
\end{definition}
We are now ready to define the team semantics for modal logic and modal logic with intuitionistic disjunction. Similar to the case of propositional logic, the team semantics of modal logic, in a rather strong sense, coincides with the traditional semantics of modal logic defined via pointed Kripke models.
\begin{definition}
Let $\mathrm{K}$ be a Kripke model. The satisfaction relation $\mathrm{K},T\models \varphi$ for $\ML$ is defined as follows. 
\begin{align*}
\mathrm{K},T\models p  \quad\Leftrightarrow\quad& w\in V(p) \, \text{ for every $w\in T$.}\\
\mathrm{K},T\models \neg p \quad\Leftrightarrow\quad& w\not\in V(p) \, \text{ for every $w\in T$.}\\
\mathrm{K},T\models (\varphi\land\psi) \quad\Leftrightarrow\quad& \mathrm{K},T\models\varphi \text{ and } K,T\models\psi.\\
\mathrm{K},T\models (\varphi\lor\psi) \quad\Leftrightarrow\quad& \mathrm{K},T_1\models\varphi \text{ and } 
\mathrm{K},T_2\models\psi \,
\text{ for some $T_1,T_2$ such that $T_1\cup T_2= T$}.\\
\mathrm{K},T\models \Diamond\varphi \quad\Leftrightarrow\quad& \mathrm{K},T'\models\varphi \text{ for some $T'$ such that $T[R]T'$}.\\
\mathrm{K},T\models \Box\varphi \quad\Leftrightarrow\quad& \mathrm{K},T'\models\varphi, \text{ where $T'=R[T]$}.\\
\intertext{For $\ML(\varovee)$ we have the following additional clause:}
\mathrm{K},T\models (\varphi\varovee\psi) \quad\Leftrightarrow\quad& \mathrm{K},T\models\varphi \text{ or } \mathrm{K},T\models\psi.
\end{align*}
\end{definition}
\begin{proposition}[\cite{Sevenster:2009}]
Let $\varphi\in\ML$, $\mathrm{K}$ be a Kripke model and $T$ a team of $\mathrm{K}$. Then
\[
\mathrm{K},T\models \varphi \quad\text{ iff }\quad \forall w\in T:\mathrm{K},w\models_{\ML} \varphi.
\]
Here $\models_{\ML}$ refers to the ordinary satisfaction relation of modal logic defined via pointed Kripke models. 
\end{proposition}

The syntax for \emph{modal dependence logic} $\MDL(\Phi)$ is obtained by extending the syntax of $\ML(\Phi)$ by propositional dependence atoms
\[
\varphi\ddfn\dep{p_1,\dots, p_n,q},
\]
where $p_1,\dots,p_n,q\in \Phi$, whereas the syntax for \emph{extended modal dependence logic} $\EMDL(\Phi)$ is obtained by extending the syntax of $\ML(\Phi)$ by \emph{modal dependence atoms}
\[
\varphi\ddfn\dep{\varphi_1,\dots, \varphi_n,\psi},
\]
where $\varphi_1,\dots,\varphi_n,\psi$ are $\ML(\Phi)$-formulae. 

The intuitive meaning of the modal dependence atom $\dep{\varphi_1,\dots, \varphi_n,\psi}$ is that the truth value of the formula $\psi$ is completely determined by the truth values of the formulae $\varphi_1,\dots, \varphi_n$.
The semantics for these dependence atoms is defined as follows.
\begin{align*}
\mathrm{K},T\models \dep{\varphi_1,\dots,\varphi_n,\psi} \quad\Leftrightarrow\quad& \forall w,v\in T: \bigwedge_{i=1}^{n}(\mathrm{K},\{w\}\models\varphi_i \Leftrightarrow \mathrm{K},\{v\}\models\varphi_i)\\
& \text{implies }(\mathrm{K},\{w\}\models\psi\Leftrightarrow \mathrm{K},\{v\}\models\psi).
\end{align*}

The following proposition for $\MDL$ and $\ML(\varovee)$ is due to \cite{va08} and \cite{ebloya}, respectively. For $\EMDL$ it follows by the fact that $\EMDL$ translates into $\ML(\varovee)$, see \cite{EHMMVV13}.
\begin{proposition}[Downwards closure]\label{dcml}
Let $\varphi$ be a formula of $\ML(\varovee)$ or $\EMDL$, let $\mathrm{K}$ be a Kripke model and let $S\subseteq T$ be teams of $\mathrm{K}$. Then $\mathrm{K},T\models \varphi$ implies $\mathrm{K},S\models \varphi$.
\end{proposition}
The standard concept of bisimulation from modal logic can be lifted, in a straightforward manner, to handle team semantics. Below when stating that $\mathrm{K},w$ and $\mathrm{K},w'$ are bisimilar, we refer to the standard bisimulation of modal logic, for a definition see, e.g., \cite{blackburn}.
\begin{definition}
Let $\mathrm{K}$ and $\mathrm{K}'$ be Kripke models and let $T$ and $T'$ be teams of $\mathrm{K}$ and $\mathrm{K}'$, respectively. 
We say that $\mathrm{K},T$ and $\mathrm{K}',T'$ are \emph{team bisimilar} if
\begin{enumerate}
\item for every $w\in T$ there exists some $w'\in T'$ such that $\mathrm{K},w$ and $\mathrm{K}',w'$ are bisimilar, and
\item for every $w'\in T'$ there exists some $w\in T$ such that $\mathrm{K},w$ and $\mathrm{K}',w'$ are bisimilar.
\end{enumerate}
\end{definition}

\begin{theorem}[\cite{HeLuSaVi14}]
If $\mathrm{K},T$ and $\mathrm{K}',T'$ are team bisimilar, then for every formula $\varphi\in\ML(\varovee)$ (and also for every $\varphi\in\EMDL$)
\[
\mathrm{K},T\models \varphi \quad\Leftrightarrow\quad  \mathrm{K}',T'\models \varphi.
\]
\end{theorem}
The following result is stated in \cite{fanthesis}. It also follows by a direct team bisimulation argument.
\begin{corollary}\label{disjoint unions}
Truth of $\ML(\varovee)$-formulae is preserved under taking disjoint unions, i.e., if $\mathrm{K}$ and $\mathrm{K}'$ are Kripke models, $T$ is a team of $\mathrm{K}$ and $\mathrm{K}\uplus \mathrm{K}'$ denotes the disjoint union of $\mathrm{K}$ and $\mathrm{K}'$ then
\[
\mathrm{K},T\models \varphi \quad\Leftrightarrow\quad \mathrm{K}\uplus \mathrm{K}',T\models \varphi,
\]
for every $\varphi\in\ML(\varovee)$.
\end{corollary}
%
%
\section{Dependency quantified Boolean formulae}\label{dqbf}
Deciding whether a given quantified Boolean formula is true is a canonical $\PSPACE$-complete problem. Dependency quantified Boolean formulae introduced by Peterson et~al. \cite{Peterson2001} are variants of quantified Boolean formulae for which the corresponding decision problem is $\NEXPTIME$-complete. In this section we give a definition of quantified Boolean formulae and dependency quantified Boolean formulae suitable for our needs. 

A \emph{Boolean variable} is a variable that is assigned either true or false. Let $\mathrm{BVAR}=\{\gamma_i \mid i\in \mathbb{Z}_+\}$ be the set of exactly all Boolean variables. \emph{Boolean formulae} $\varphi$ are a built from Boolean variables by the following grammar:
\[
\varphi \ddfn \alpha \mid \neg\alpha \mid (\varphi \wedge \varphi) \mid (\varphi \vee \varphi),
\] 
where $\alpha \in \mathrm{BVAR}$.
A formula
\[
\psi = Q_1\alpha_1 Q_2 \alpha_2\dots Q_n\alpha_{n}\varphi,
\]
where $Q_i\in \{\forall,\exists\}$, for each $i\leq n$, is called a \emph{quantified Boolean formula}, if $\varphi$ is a Boolean formula and $\psi$ does not have free variables. We let $\QBF$ denote the set of all quantified Boolean formulae. Semantics for Boolean formulae and quantified Boolean formulae is defined via assignments $s:\mathrm{BVAR}\to \{0,1\}$ in the obvious way.
We define that
\[
\TQBF=\{\varphi\in \QBF \mid \varphi \text{ is true}\}.
\]
\begin{theorem}[\cite{Stockmeyer:1973}]
The membership problem of $\TQBF$ is $\PSPACE$-complete.
\end{theorem}
We call a formula
\[
\psi = \forall \alpha_1 \dots \forall\alpha_{n}\exists \beta_1\dots \exists \beta_k\varphi
\]
a \emph{simple quantified Boolean formula}, if $\varphi$ is a Boolean formula, $\psi$ does not have free variables and each variable quantified in $\psi$ is quantified exactly once. Let $P_1,\dots,P_k\subseteq \{\alpha_1,\dots,\alpha_n\}$. We call the tuple $(P_1,\dots,P_k)$ a \emph{constraint} for $\psi$. If $P_1\subseteq P_2\subseteq\dots \subseteq P_k$, we call the constraint \emph{simple}. The idea here is that, for each $i\leq k$, the value assigned for the existentially quantified Boolean variable $\beta_i$ may only depend on the values given to the universally quantified Boolean variables in the set $P_i$. Thus, the intuition is that the simple quantified Boolean formula
\[
\forall \alpha_1 \forall \alpha_2 \exists\beta_1 \exists \beta_2 \theta
\]
is true under the constraint $(\{\alpha_1\}, \{\alpha_2\})$, if $\theta$ can be made true such that the dependencies $\dep{\alpha_1,\beta_1}$ and $\dep{\alpha_2,\beta_2}$ hold. The formal definition is given below.
\begin{definition}
Let $\psi = \forall \alpha_1 \dots \forall\alpha_{n}\exists \beta_1\dots \exists \beta_k\varphi$ be a simple quantified Boolean formula and $(P_1,\dots,P_k)$ a constraint for $\psi$. We say that $\psi$ is \emph{true under the constraint} $(P_1,\dots,P_k)$, if there exists a function $f_i:\{0,1\}^{\lvert P_i\rvert}\to \{0,1\}$, for each $i\leq k$, such that for each assignment $s:\{\alpha_1,\dots,\alpha_n\}\to\{0,1\}$
\[
s'\models \varphi,
\]
where $s'$ is the modified assignment defined as follows: 
\[
s'(\alpha):=
\begin{cases}
f_i\big(s(P_i)\big) &\text{if } \alpha = \beta_i \text{ and } i\leq k,\\
s(\alpha) &\mbox{otherwise.}
\end{cases}
\]
Here $s(P_i)$ is a shorthand notation for $\big(s(\gamma_{i_1}), \dots, s(\gamma_{i_t})\big)$, where $\gamma_{i_1}, \dots, \gamma_{i_t}$ are exactly the Boolean variables in $P_i$ ordered such that $i_j<i_{j+1}$, for each $j< t$.
\end{definition}
It is easy to see that there is a close connection between quantified Boolean formulae and simple quantified Boolean formulae with simple constraints; there exists a polynomial time computable function $F$ that associates each quantified Boolean formula to an equivalent simple quantified Boolean formula with a simple constraint, and vice versa. The equivalent quantified Boolean formula is obtained from a simple quantified Boolean formula with a simple constraint by reordering the quantification of variables. The constraint determines the order of quantifiers.

We define that a dependency quantified Boolean formula is a pair $(\psi, \vec P)$ where $\psi$ is a simple quantified Boolean formula and $\vec{P}$ is a constraint for $\psi$. We let $\DQBF$ denote the set of all dependency quantified Boolean formulae. We define that
\[
\TDQBF=\{(\psi, \vec P)\in \DQBF \mid \psi\text{ is true under the constraint }\vec P\}.
\]
\begin{theorem}[\cite{Peterson2001}]\label{TDQBF problem}
The membership problem of $\TDQBF$ is $\NEXPTIME$-complete.
\end{theorem}
%
%
%
\section{Computational complexity of propositional dependence logics}\label{ccpd}
%
Computational complexity of the satisfiability problem and the model checking problem for variants of propositional and modal dependence logics have been thoroughly studied, see e.g., \cite{Sevenster:2009, EbbingL12, EHMMVV13, lohvo13, Graedel:2013, fanthesis}. However there is not much research done on the validity problem of these logics. Note that since the logics of dependence are not closed under negation the traditional connection between the satisfiability problem and the validity problem fails.
\subsection{Satisfiability, validity and model checking in team semantics}
We start by defining satisfiability and validity in the context of team semantics.

A formula $\varphi$ of propositional dependence logic is said to be \emph{satisfiable}, if there exists a propositional team $X$ such that $X\models\varphi$.
A formula $\varphi$ of propositional dependence logic is said to be \emph{valid}, if $X\models\varphi$ holds for all teams $X$ such that the proposition symbols of $\varphi$ are in the domain of $X$.
Analogously, a formula $\psi$ of $\EMDL$ \big(or $\ML(\varovee)$\big) is said to be \emph{satisfiable}, if there exists a Kripke model $\mathrm{K}$ and a team $T$ of $\mathrm{K}$ such that $\mathrm{K},T\models\psi$.
A formula $\psi$ of $\EMDL$ \big(or $\ML(\varovee)$\big) is said to be \emph{valid}, if $\mathrm{K},T\models\psi$ holds for every Kripke model $\mathrm{K}$ (such that the proposition symbols in $\psi$ are mapped by the valuation of $\mathrm{K}$) and every team $T$ of $\mathrm{K}$.

The satisfiability problem and the validity problem for these logics is defined in the obvious manner. Given a binary encoding of a formula of a given logic, decide whether the formula is satisfiable (valid, respectively). The variant of the model checking problem, we are concerned in this article is the following.
Given binary encodings of a formula $\varphi$ of propositional dependence logic and of a (finite) propositional team $X$, decide whether $X\models\varphi$. The corresponding problem for modal logics is defined as follows. Given binary encodings of a formula $\psi$ of $\EMDL$ \big(or $\ML(\varovee)$\big), of a finite Kripke model $\mathrm{K}$ and of a team $T$ of $\mathrm{K}$, decide whether $\mathrm{K},T\models\psi$.

\subsection{The validity problem of propositional dependence logic}\label{valpd}
The complexity of the satisfiability problem for $\PL$ and $\PD$ is known to coincide; both are $\NP$-complete. The result for $\PL$ is due to Cook \cite{co71} and Levin \cite{Lev1973}. For $\PD$, the $\NP$-hardness follows directly from the result of Cook and Levin, and the inclusion to $\NP$ follows from the work of Lohmann and Vollmer \cite{lohvo13}. 

A natural question then arises: Is there a similar connection between the validity problem of $\PL$ and that of $\PD$?
Since the syntax of propositional logic is closed under taking negations, it follows that the validity problem for $\PL$ is $co\NP$-complete. However, since the syntax of propositional dependence logic is not closed under taking negations, the corresponding connection between the satisfiability problem and the validity problem of $\PD$ fails. This indicates that there might not be any direct connection between the validity problem of $\PL$ and that of $\PD$. In fact, as we will see, the validity problem for $\PD$ is much harder than the corresponding problem for $\PL$.
Surprisingly, we are able to show that the validity problem for $\PD$ is $\NEXPTIME$-complete.

We shall first show that the validity problem for $\PD$ is in $\NEXPTIME$. To that end, we use the following result concerning the model checking problem of $\PD$.

\begin{theorem}[\cite{EbbingL12}]\label{pdmcnp}
The model checking problem for $\PD$ is $\NP$-complete.
\end{theorem}

Let $D$ be a finite set of proposition symbols. By $X_{\max D}$ we denote the set of all assignments $s:D \to \{0,1\}$. The following lemma follows directly from the fact that $\PD$ is downward closed, i.e., Proposition \ref{dcprop}.
\begin{lemma}\label{bigteam}
Let $\varphi$ be a formula of $\PD$ and let $D$ be the set of proposition symbols occurring in $\varphi$. Then
$\varphi$ is valid if and only if $X_{\max D}\models \varphi$.
\end{lemma}
\begin{lemma}\label{pdvalin}
The validity problem for $\PD$ is in $\NEXPTIME$.
\end{lemma}
\begin{proof}
Let $\varphi$ be a $\PD$-formula. Let $D$ be the set of proposition symbols occurring in $\varphi$. Now, by Lemma \ref{bigteam}, $\varphi$ is valid if and only if $X_{\max D}\models \varphi$. The size of $X_{\max D}$ is $2^{\lVert D\rVert}$ and thus $\leq 2^{\lVert \varphi\rVert}$. Therefore $X_{\max D}$ can be clearly constructed from $\varphi$ in exponential time. By Theorem \ref{pdmcnp}, there exists an $\NP$ algorithm (with respect to $\lVert X_{\max D}\rVert + \lVert \varphi\rVert$) for checking whether $X_{\max D}\models \varphi$. Clearly this algorithm works in $\NEXPTIME$ with respect to the size of $\varphi$. Therefore, we conclude that the validity problem for $\PD$ is in $\NEXPTIME$.
\end{proof}
We will then show that the validity problem for $\PD$ is $\NEXPTIME$-hard. We give a reduction from $\TDQBF$ to the validity problem of $\PD$.
\begin{lemma}\label{pdvalhard}
The validity problem for $\PD$ is $\NEXPTIME$-hard.
\end{lemma}
\begin{proof}
We will give a reduction from the truth problem of dependency quantified Boolean formulae to the validity problem of $\PD$. Since Boolean variables and proposition symbols in the context of $\PD$ are essentially the same, we will in this proof treat Boolean variables as proposition symbols, and vice versa.
Consequently, we may treat quantifier free Boolean formulae as formulae of propositional logic, and vice versa.

We will associate each $\DQBF$-formula $\mu$ with a corresponding $\PD$ formula $\varphi_\mu$. Let
\[
\mu=\big(\forall \alpha_1\dots \forall \alpha_n\exists \beta_1\dots\exists \beta_k\, \psi, (P_1,\dots,P_k)\big)
\]
be a $\DQBF$-formula. For each set of Boolean variables $P_i$, $i\leq k$, we stipulate that
\(
P_i=\{\alpha_{i_1},\dots,\alpha_{i_{n_i}}\}.
\)
We then denote by $D_\mu$ the set of Boolean variables in $\mu$, i.e., $D_\mu \dfn \{\alpha_1,\dots, \alpha_n,\beta_1,\dots,\beta_k\}$. Recall that we treat Boolean variables also as proposition symbols. Let
%
\[
\varphi_\mu \dfn \psi \vee \bigvee_{i\leq k} \dep{\alpha_{i_1},\dots,\alpha_{i_{n_i}}, \beta_i}.
\]
We will show that $\mu$ is true (i.e., $\mu\in \TDQBF$) if and only if the corresponding $\PD$-formula $\varphi_\mu$ is valid. Since $\TDQBF$ is $\NEXPTIME$-complete and $\varphi_\mu$ is polynomial with respect to $\mu$, it follows that the validity problem for $\PD$ is $\NEXPTIME$-hard. By Lemma \ref{bigteam}, it is enough to show that $\mu$ is true if and only if $X_{\max D_\mu}\models \varphi_\mu$.

Assume first that $\mu$ is true, i.e., that $\forall \alpha_1\dots \forall \alpha_n\exists \beta_1\dots\exists \beta_k\, \psi$ is true under the constraint $(P_1,\dots,P_k)$. Therefore, for each $i\leq k$, there exists a function $f_i:\{0,1\}^{\lvert P_i\rvert}\to \{0,1\}$ such that
\begin{equation}\label{eq:1}
\text{for every assignment }s:\{\alpha_1, \dots, \alpha_n\}\to\{0,1\}:\quad s'\models \psi,
\end{equation}
where $s'$ is the modified assignment defined as follows: 
\[
s'(\alpha):=
\begin{cases}
f_i\big(s(P_i)\big) &\text{if } \alpha = \beta_i \text{ and } i\leq k,\\
s(\alpha) &\mbox{otherwise.}
\end{cases}
\]
Our goal is to show that
\[
X_{\max D_\mu}\models\psi \vee \bigvee_{i\leq k} \dep{\alpha_{i_1},\dots,\alpha_{i_{n_i}}, \beta_i}.
\]
It suffices to show that there exist some $Y,Z_1,\dots Z_k\subseteq X_{\max D_\mu}$ such that $Y\cup Z_1\cup\dots\cup Z_k = X_{\max D_\mu}$, $Y\models \psi$, and $Z_i\models\dep{\alpha_{i_1},\dots,\alpha_{i_{n_i}}, \beta_i}$, for each $i\leq k$.
We define the team $Z_i$, for each $i\leq k$, by using the function $f_i$.
We let
\(
Z_i \dfn \{s\in X_{\max D_\mu} \mid s(\beta_i)\neq f_i\big(s(\alpha_{i_1}),\dots, s(\alpha_{i_{n_i}})\big)\},
\)
for each $i\leq k$. Now, since Boolean variables have only $2$ possible values, we conclude that, for each $i\leq k$, $Z_i\models \dep{\alpha_{i_1},\dots,\alpha_{i_{n_i}}, \beta_i}$. Thus
\begin{equation}\label{eq:2}
\bigcup_{1\leq i\leq k} Z_i \models \bigvee_{i\leq k} \dep{\alpha_{i_1},\dots,\alpha_{i_{n_i}}, \beta_i}.
\end{equation}
Note that $s(\beta_i)= f_i\big(s(\alpha_{i_1}),\dots, s(\alpha_{i_{n_i}})\big)$ holds for every $s\in (X_{\max D_\mu}\setminus Z_i)$ and every $i\leq k$. 
Define then that
\[
Y\dfn X_{\max D_\mu}\setminus \bigcup_{1\leq i\leq k} Z_i.
\]
Clearly, for every $s\in Y$ and $i\leq k$, it holds that $s(\beta_i)= f_i\big(s(\alpha_{i_1}),\dots, s(\alpha_{i_{n_i}})\big)$. Thus from \eqref{eq:1}, it follows that $s\models \psi$, for every $s\in Y$. Since $\psi$ is a $\PL$ formula, we conclude by Proposition \ref{PLflat} that $Y\models \psi$. From this together with \eqref{eq:2}, we conclude that $X_{\max D_\mu}\models \varphi_\mu$.

Assume then that $X_{\max D_\mu}\models \varphi_\mu$. Therefore
\[
Y\models \bigvee_{i\leq k} \dep{\alpha_{i_1},\dots,\alpha_{i_{n_i}}, \beta_i}
\]
and $Z\models \psi$, for some $Y$ and $Z$ such that $Y\cup Z=X_{\max D_\mu}$. Hence there exist some $Y_1,\dots,Y_k,Z$ such that
\(
Y_1\cup\dots\cup Y_k\cup Z=X_{\max D_\mu},
\)
$Z\models \psi$, and $Y_i\models \dep{\alpha_{i_1},\dots, \alpha_{i_{n_i}},\beta_i}$ for each $i\leq k$. Assume that we have picked $Y_1,\dots, Y_k, Z$ such that $Z$ is minimal. We will show that then $Z\models \dep{\alpha_{i_1},\dots, \alpha_{i_{n_i}},\beta_i}$, for each $i\leq k$. Assume for the sake of contradiction that, for some $i\leq k$, there exist $s,t\in Z$ such that
\[
s(\alpha_{i_1})=t(\alpha_{i_1}), \dots, s(\alpha_{i_{n_i}})=t(\alpha_{i_{n_i}}) \text{ but } s(\beta_i)\neq t(\beta_i).
\]
Now clearly either $Y_i\cup\{s\}\models \dep{\alpha_{i_1},\dots, \alpha_{i_{n_i}},\beta_i}$ or $Y_i\cup\{t\}\models \dep{\alpha_{i_1},\dots, \alpha_{i_{n_i}},\beta_i}$. This contradicts the fact that $Z$ was assumed to be minimal.

We will then show that for every $a_1,\dots,a_n\in \{0,1\}$ there exists some assignment $s$ in $Z$ that expands
\[
(\alpha_1,\dots,\alpha_n)\mapsto (a_1, \dots, a_n).
\]
Let $a_1,\dots,a_n\in \{0,1\}$.
%
%
Now, for every $i\leq k$, since
$Y_i\models \dep{\alpha_{i_1},\dots, \alpha_{i_{n_i}},\beta_i},$
it follows that for any two $s',s''\in Y_i$ that expand $(\alpha_1, \dots, \alpha_n)\mapsto (a_1,\dots,a_n)$, it holds that $s'(\beta_i)=s''(\beta_i)$. Thus, for each $i\leq k$, there exists a truth value $b_i\in\{0,1\}$ such that there is no expansions of $(\alpha_1, \dots, \alpha_n,\beta_i)\mapsto (a_1,\dots,a_n,b_i)$ in $Y_i$. Therefore, the assignment $(\alpha_1, \dots, \alpha_n,\beta_1,\dots,\beta_k)\mapsto (a_1,\dots,a_n,b_1,\dots,b_k)$ is not in $Y_i$, for any $i\leq k$. Thus the assignment  $(\alpha_1, \dots, \alpha_n,\beta_1,\dots,\beta_k)\mapsto (a_1,\dots,a_n,b_1,\dots,b_k)$ is in $Z$.
Hence, for every $a_1, \dots, a_n\in \{0,1\}$, there exists some expansion of $(\alpha_1, \dots, \alpha_n) \mapsto (a_1, \dots, a_n)$ in $Z$.

Now, for each $i\leq k$, we define the function $f_i:\{0,1\}^{\lvert P_i\rvert}\to \{0,1\}$ as follows. We define that
\[
f_i(b_1,\dots,b_{\lvert P_i\rvert}) := s(\beta_i),
\]
where $s$ is an assignment in $Z$ that expands $(\alpha_{i_1},\dots\alpha_{i_{n_i}})\mapsto (b_1,\dots,b_{\lvert P_i\rvert})$. Since $Z\models \dep{\alpha_{i_1},\dots, \alpha_{i_{n_i}},\beta_i}$, for each $i\leq k$, the functions $f_i$ are well defined.
Now since $\psi$ is syntactically a $\PL$ formula and since $Z\models \psi$, it follows from proposition \ref{PLflat} that $s'\models \psi$, for each $s'\in Z$. Clearly the functions $f_i$, for $i\leq k$, are as required in \ref{eq:1}. Thus we conclude that \ref{eq:1} holds. Thus $\mu$ is true.

Now since the truth problem for $\DQBF$ is $\NEXPTIME$-hard and $\varphi_\mu$ is clearly polynomial with respect to $\mu$, we conclude that the validity problem for $\PD$ is $\NEXPTIME$-hard. 

\end{proof}

By Lemmas \ref{pdvalin} and \ref{pdvalhard}, we obtain the following:
\begin{theorem}\label{pdc}
The validity problem for $\PD$ is $\NEXPTIME$-complete.
\end{theorem}

%
\section{The validity problem for modal dependence logic}\label{ccml}
%
%
The satisfiability problem for both $\MDL$ and $\EMDL$ is known to be $\NEXPTIME$-complete. For $\MDL$ this was shown by Sevenster \cite{Sevenster:2009} and for $\EMDL$ Ebbing et~al. \cite{EHMMVV13}.
In Theorem \ref{pdc} we showed that the validity problem for $\PD$ is $\NEXPTIME$-complete and thus much more complex than the corresponding satisfiability problem. This together with the fact that the validity problem for modal logic is known to be $\PSPACE$-complete (Laddner\cite{la77}) seems to suggest $\EXPSPACE$ as a candidate for the complexity of the validity problem of $\MDL$ and $\EMDL$. However, we manage to do a bit better.
We establish that the validity problem of $\MDL$ and $\EMDL$ is in $\NEXPTIME^\NP$, i.e., in $\NEXPTIME$ with access to $\NP$ oracles. Thus we obtain that the precise complexity of these problems lie somewhere between and $\NEXPTIME$ and $\NEXPTIME^\NP$, since the $\NEXPTIME$-hardness follows directly from Lemma \ref{pdvalhard}.
\begin{corollary}
The validity problem for $\MDL$ and $\EMDL$ is $\NEXPTIME$-hard.
\end{corollary}
The rest of this section is devoted on showing that the validity problem for $\EMDL$ is in $\NEXPTIME^\NP$.

Let $\varphi$ be a formula of $\EMDL$ or $\ML(\varovee)$. The set $\mathrm{nbSubf}(\varphi)$ of \emph{non-Boolean subformulas} of $\varphi$ is defined recursively as follows.
\begin{align*}
& \mathrm{nbSubf}(\neg p) := \mathrm{nbSubf}(p) := \{p\}, \quad \mathrm{nbSubf}(\triangle \varphi) := \{\triangle \varphi \} \cup \mathrm{nbSubf}(\varphi) \quad \text{for $\triangle \in \{\Diamond, \Box\}$},\\
& \mathrm{nbSubf}(\varphi \circ \psi) := \mathrm{nbSubf}(\varphi)\cup \mathrm{nbSubf}(\psi) \quad \text{for $\circ \in \{ \varovee, \vee, \wedge \}$},\\
%
%
& \mathrm{nbSubf}\big(\dep{\varphi_1,\dots, \varphi_n,\psi}\big) := \mathrm{nbSubf}(\varphi_1)\cup\dots\cup \mathrm{nbSubf}(\varphi_n)\cup \mathrm{nbSubf}(\psi).
\end{align*}
The following lemma follows directly from \cite[Claim 15]{Sevenster:2009}.
\begin{lemma}\label{smallmodels}
Let $\varphi\in \ML$ and let $k=\lvert \mathrm{nbSubf}(\varphi)\rvert$.
Then, $\varphi$ is valid if and only if $\mathrm{K},w\models \varphi$ holds for every Kripke model $\mathrm{K}= (W,R,V)$ and $w\in W$ such that $\lvert W\rvert \leq 2^k$.
\end{lemma}
The following proposition for $\EMDL$ is based on a similar result for $\MDL$ that essentially combines the ideas of \cite{va08}, \cite{Sevenster:2009} and \cite{lohmannthesis}.
\begin{proposition}\label{EMDLtrans}
For every formula $\varphi\in \EMDL$ there exists an equivalent formula
\[
\varphi^* = \Idis_{i\in I}\varphi_i,
\]
where $I$ is a finite set of indices and $\varphi_i\in \ML$, for each $i\in I$. Furthermore, for each $i\in I$, the size of $\varphi_i$ is only exponential in the size of $\varphi$ and $\lvert\mathrm{nbSubf}(\varphi_i)\rvert \leq 3\times\lvert \varphi \rvert$.
\end{proposition}
\begin{proof}
We will first recall an exponential translation $\varphi\mapsto\varphi^+$ from $\EMDL$ to $\ML(\varovee)$ given in \cite[Theorem 2]{EHMMVV13}. The cases for proposition symbols, Boolean connectives and modalities are trivial, i.e.,
\(
p \mapsto p, \neg p \mapsto \neg p, (\varphi\wedge\psi) \mapsto (\varphi^+\wedge\psi^+), (\varphi\vee\psi) \mapsto (\varphi^+\vee\psi^+), \Diamond\varphi \mapsto \Diamond\varphi^+, \Box\varphi \mapsto \Box\varphi^+.
\)
The only interesting case is the case for the dependence atom. We define that
\[
\dep{\varphi_1,\dots,\varphi_n,\psi} \quad\mapsto\quad \bigvee_{a_1,\dots,a_n\in \{\bot,\top\}}\big(\bigwedge_{i\leq n} \varphi_i^{a_i}\wedge(\psi\varovee\psi^\bot)\big), 
\]
where $\varphi^\top$ denotes $\varphi$ and $\varphi^\bot$ denotes the $\ML$ formula obtained from $\neg\varphi$ by pulling all negations to the atomic level. Notice that the size of $\varphi^+$ is $\leq c\times \lvert\varphi\rvert\times 2^{\lvert\varphi\rvert}$, for some constant $c$. Thus the size of $\varphi^+$ is at most exponential with respect to the size of $\varphi$. From $\varphi^+$ it is easy to obtain an equivalent $\ML(\varovee)$-formula $\varphi^*$ of the form
\[
\Idis_{i\in I}\varphi_i,
\]
where $I$ is a finite index set and $\varphi_i$, for $i\in I$, is an $\ML$-formula. Let $F$ be the set of all selection functions $f$ that select, separately for each occurrence, either the left disjunct $\psi$ or the right disjunct $\theta$ of each subformula of the form $(\psi\varovee\theta)$ of $\varphi^+$. Now let $\varphi^+_f$ denote the formula obtained from $\varphi^+$ by substituting each occurrence of a subformula of type $(\psi\varovee\theta)$ by $f\big((\psi\varovee\theta)\big)$. We then define that
\[
\varphi^*\dfn \Idis_{f\in F} \varphi^+_f.
\]
It is straightforward to prove that $\varphi^*$ is equivalent to $\varphi^+$ and hence to $\varphi$. 
Since, for each $f\in F$, $\varphi^+_f$ is obtained from $\varphi^+$ by substituting subformulae of type $(\psi\varovee\theta)$ with either $\psi$ or $\theta$, it is clear that the size of $\varphi^+_f$ is bounded above by the size of $\varphi^+$. Recall that the size of $\varphi^+$ is at most exponential with respect to the size of $\varphi$. 
Therefore, for each $f\in F$, the size of $\varphi^+_f$ is at most exponential with respect to the size of $\varphi$.

We say that the modal operator $\Diamond$ in $\Diamond \theta$ \emph{dominates} an intuitionistic disjunction if $\varovee$ occurs in $\theta$.
To see that $\vert\mathrm{nbSubf}(\varphi^+_f)\rvert\leq 3\times\lvert \varphi \rvert$, for each $f\in F$, notice first that in the translation $\varphi \mapsto \varphi^+$ the only case that can increase the number of non-Boolean subformulae is the case for the dependence atom. Each $\varphi_i^{\bot}$ and $\psi^{\bot}$ may introduce new non-Boolean subformulae. Thus it is straightforward to see that $\lvert \mathrm{nbSubf}(\varphi^+)\rvert \leq 2\times\lvert \mathrm{nbSubf}(\varphi)\rvert$. 
Furthermore, notice that the number of modal operators that dominate an intuitionistic disjunction in $\varphi^+$ is less or equal to the number of modal operators in $\varphi$. Let $k$ denote the number of modal operators in $\varphi$. It is easy to see that $\lvert\mathrm{nbSubf}(\varphi^+_f)\rvert\leq \lvert \mathrm{nbSubf}(\varphi^+) \rvert + k$, for each $f\in F$. Now since $k\leq \lvert \varphi \rvert$ and $\lvert \mathrm{nbSubf}(\varphi)\rvert \leq \lvert \varphi \rvert$, we obtain that $\lvert\mathrm{nbSubf}(\varphi^+_f)\rvert\leq 3\times\lvert \varphi \rvert$, for each $f\in F$. With a more careful bookkeeping, we would obtain that $\mathrm{nbSubf}(\varphi^+_f)\leq 2\times\lvert \varphi \rvert$.
\end{proof}
We say that a formula $\varphi\in \ML$ is \emph{valid in small models} if $\mathrm{K},w\models \varphi$ holds for every Kripke model $\mathrm{K}=(W,R,V)$ and $w\in W$ such that $\lvert W \rvert \leq \lvert \varphi \rvert$.
\begin{lemma}
The decision problem whether a given formula of $\ML$ is valid in small models is in $co\NP$.
\end{lemma}
\begin{proof}
If a formula $\varphi\in \ML$ is not valid in small models, then there is some $k\leq \lvert \varphi\rvert$ and a pointed Kripke model $\mathrm{K},w$ of size $k$ such that $\mathrm{K},w\not\models \varphi$. The size of $\mathrm{K},w$ is clearly polynomial in $\lvert \varphi\rvert$, and thus it can be guessed nondeterministically in polynomial time with respect to $\lvert \varphi\rvert$. The model checking problem for modal logic is in $\PTIME$ (\cite{clarke86}), and thus $\mathrm{K},w\not\models \varphi$ can be verified in polynomial time with respect to $\lvert \mathrm{K}\rvert + \lvert \varphi \rvert$ and thus in polynomial time with respect to $\lvert \varphi \rvert$.
\end{proof}
\begin{proposition}\label{disjunctionprop}
$\ML(\varovee)$ has the $\varovee$-disjunction property, i.e.,
for every $\varphi,\psi\in \ML(\varovee)$ it holds that
$(\varphi\varovee\psi)$ is valid if and only if either $\varphi$ is valid or $\psi$ is valid.
\end{proposition}
\begin{proof}
%
The direction from right to left is trivial. We will prove here the direction form left to right. Assume that $(\varphi\varovee\psi)$ is valid. For the sake of contradiction, assume then that neither $\varphi$ nor $\psi$ is valid. Thus there exist Kripke models $\mathrm{K}$ and $\mathrm{K'}$, and teams $T$ and $T'$ of $\mathrm{K}$ and $\mathrm{K'}$, respectively, such that $\mathrm{K},T\not\models \varphi$ and $\mathrm{K'},T'\not\models \psi$. From Corollary \ref{disjoint unions} it follows that $\mathrm{K}\uplus\mathrm{K'},T\not\models \varphi$ and $\mathrm{K}\uplus\mathrm{K'},T'\not\models \psi$, where $\mathrm{K}\uplus\mathrm{K'}$ denotes the disjoint union of $\mathrm{K}$ and $\mathrm{K'}$. Since the formulae of $\ML(\varovee)$ are downwards closed (Proposition  \ref{dcml}), we conclude that $\mathrm{K}\uplus\mathrm{K'},T\cup T'\not\models \varphi$ and $\mathrm{K}\uplus\mathrm{K'},T\cup T'\not\models \psi$. Thus $\mathrm{K}\uplus\mathrm{K'},T\cup T'\not\models (\varphi\varovee\psi)$. This contradicts the fact that $(\varphi\varovee\psi)$ is valid.
\end{proof}
%
\begin{proposition}\label{valemdl}
The validity problem for $\EMDL$ is in $\NEXPTIME^\NP$.
\end{proposition}
\begin{proof}
For deciding whether a given $\EMDL$ formula is valid, we give a nondeterministic exponential time algorithm that has an access to an $\NP$ oracle that decides whether a given $\ML$ formula is valid in small models.
For each $\varphi\in\EMDL$ let $\varphi^+$ denote the equivalent exponential size $\ML(\varovee)$-formula from the proof of Proposition \ref{EMDLtrans}. Clearly $\varphi^+$ is computable from $\varphi$ in exponential time. Furthermore let $\varphi^*$ denote the $\ML(\varovee)$-formula of the form
\(
\Idis_{f\in F}\varphi_f^+
\)
of Proposition \ref{EMDLtrans}. Moreover let $g:\mathbb{N}\to\mathbb{N}$ be some exponential function such that $\lvert\varphi_f^+\rvert \leq g(\lvert\varphi\rvert)$, for every $\varphi\in\EMDL$ and $f\in F$. By Proposition \ref{EMDLtrans} there exists such a function.
%
%


We are now ready to give a $\NEXPTIME^\NP$ algorithm for the validity problem of $\EMDL$. Let $\varphi$ be an $\EMDL$ formula. First guess nondeterministically an $\ML$ formula $\psi$ of the same vocabulary as $\varphi$ of size at most $g(\lvert\varphi\rvert)$. Then compute $\varphi^+$ from $\varphi$ and check whether $\psi$ is among the disjuncts $\varphi^+_f$, $f\in F$, of $\varphi^*$.
Clearly the checking can be done in polynomial time with respect to $\lvert\varphi^+\rvert + \lvert\psi\rvert$ and thus in exponential time with respect to the size of $\varphi$.
If $\psi$ is not among the disjuncts the algorithm outputs ``No'', otherwise the algorithm continues. 
We then give
\[
\psi^-:=\big(\bigwedge_{i\leq 2^{3\times \lvert\varphi\rvert}} (p\vee \neg p) \big)\wedge\psi
\]
as an input to an $\NP$ oracle that decides whether the $\ML$ formula $\psi^-$ is valid in small models. Clearly $\psi^-$ is computable from $\psi$ in exponential time with respect to the size of $\varphi$. The algorithm outputs ``No'' if the oracle outputs ``No'' and ``Yes'' if the oracle outputs ``Yes''. Clearly this algorithm is in $\NEXPTIME^\NP$.

Now by Proposition \ref{EMDLtrans}, $\varphi$ is valid if and only if $\varphi^*$ is valid, and furthermore, by Proposition \ref{disjunctionprop}, $\varphi^*$ is valid if and only if $\varphi^+_f$ is valid for some $f\in F$. By Proposition \ref{EMDLtrans}, $\lvert \mathrm{nbSubf}(\varphi^+_f)\rvert \leq 3\times \lvert \varphi \rvert$, for every $f\in F$. Thus by Lemma \ref{smallmodels}, for every $f\in F$, $\varphi^+_f$ is valid if and only if $\varphi^+_f$ is true on all pointed models of size at most $2^{3\times\lvert \varphi \rvert}$.
Now clearly, for every $f\in F$, $\varphi^+_f$ is valid if and only if the formula
\[
\varphi^-_f :=\big(\bigwedge_{i\leq 2^{3\times\lvert\varphi\rvert}} (p\vee \neg p) \big)\wedge\varphi^+_f
\]
is valid.
Thus, for every $f\in F$, $\varphi^+_f$ is valid if and only if $\varphi^-_f$ is valid in small models.
Therefore, and since $\psi^-=\varphi^-_f$, for some $f\in F$, the algorithm decides the validity problem of $\EMDL$.
\end{proof}
\begin{corollary}
The validity problem for $\MDL$ is in $\NEXPTIME^\NP$.
\end{corollary}

\section{Conclusion}
In this article we studied the validity problem of propositional dependence logic, modal dependence logic, and extended modal dependence logic. We established that the validity problem for propositional dependence logic is $\NEXPTIME$-complete. In addition we showed that the corresponding problem for modal dependence logic and extended modal dependence logic is $\NEXPTIME$-hard and contained in $\NEXPTIME^\NP$. The exact complexity of the validity problem for $\MDL$ and $\EMDL$ remain open. We conjecture that both of these problems are harder than $\NEXPTIME$. We also believe that the complexity of $\MDL$ and $\EMDL$ coincide. In addition to resolving the precise complexity of the validity problem of $\MDL$ and $\EMDL$, we are interested in the complexity of the entailment problem of $\PD$, $\MDL$, and $\EMDL$. Note that in the context of dependence logic the entailment problem cannot be reduced directly to the validity problem. However the validity problem can be reduced to the entailment problem. Hence the entailment problem of $\PD$, $\MDL$, and $\EMDL$ is at least as hard as the corresponding validity problem.
%

\section*{Acknowledgements}
The author would like to thank the anonymous reviewers for their detailed and constructive comments. The author would also like to thank the Academy of Finland (grant 266260) and the Finnish Academy of Science and Letters for financial support.

\bibliographystyle{eptcs}

\begin{thebibliography}{99}
\providecommand{\bibitemdeclare}[2]{}
\providecommand{\surnamestart}{}
\providecommand{\surnameend}{}
\providecommand{\urlprefix}{Available at }
\providecommand{\url}[1]{\texttt{#1}}
\providecommand{\href}[2]{\texttt{#2}}
\providecommand{\urlalt}[2]{\href{#1}{#2}}
\providecommand{\doi}[1]{doi:\urlalt{http://dx.doi.org/#1}{#1}}
\providecommand{\bibinfo}[2]{#2}

\bibitemdeclare{book}{blackburn}
\bibitem{blackburn}
\bibinfo{author}{Patrick \surnamestart Blackburn\surnameend},
  \bibinfo{author}{Maarten \surnamestart de~Rijke\surnameend} \&
  \bibinfo{author}{Yde \surnamestart Venema\surnameend} (\bibinfo{year}{2001}):
  \emph{\bibinfo{title}{Modal Logic}}.
\newblock \bibinfo{publisher}{Cambridge University Press},
  \bibinfo{address}{New York, NY, USA},
	\doi{10.1017/CBO9781107050884}.

\bibitemdeclare{article}{blassgure}
\bibitem{blassgure}
\bibinfo{author}{Andreas \surnamestart Blass\surnameend} \&
  \bibinfo{author}{Yuri \surnamestart Gurevich\surnameend}
  (\bibinfo{year}{1986}): \emph{\bibinfo{title}{Henkin quantifiers and complete
  problems}}.
\newblock {\sl \bibinfo{journal}{Annals of Pure and Applied Logic}}
  \bibinfo{volume}{32}, pp. \bibinfo{pages}{1 -- 16},
  \doi{10.1016/0168-0072(86)90040-0}.

\bibitemdeclare{article}{Ciardelli11}
\bibitem{Ciardelli11}
\bibinfo{author}{Ivano \surnamestart Ciardelli\surnameend} \&
  \bibinfo{author}{Floris \surnamestart Roelofsen\surnameend}
  (\bibinfo{year}{2011}): \emph{\bibinfo{title}{Inquisitive Logic}}.
\newblock {\sl \bibinfo{journal}{Journal of Philosophical Logic}}
  \bibinfo{volume}{40}(\bibinfo{number}{1}), \doi{10.1007/s10992-010-9142-6}.

\bibitemdeclare{article}{clarke86}
\bibitem{clarke86}
\bibinfo{author}{Edmund M. \surnamestart Clarke\surnameend}, 
\bibinfo{author}{E. Allen \surnamestart Emerson\surnameend}
\&
  \bibinfo{author}{A. Prasad \surnamestart Sistla\surnameend}
  (\bibinfo{year}{1986}): \emph{\bibinfo{title}{Automatic verification of finite-state concurrent systems using temporal logic specifications}}.
\newblock {\sl \bibinfo{journal}{ACM T. Progr. Lang. Sys.}}
  \bibinfo{volume}{8}(\bibinfo{number}{2}), \doi{10.1145/5397.5399}.


\bibitemdeclare{inproceedings}{co71}
\bibitem{co71}
\bibinfo{author}{Stephen~A. \surnamestart Cook\surnameend}
  (\bibinfo{year}{1971}): \emph{\bibinfo{title}{The complexity of
  theorem-proving procedures}}.
\newblock In: {\sl \bibinfo{booktitle}{Proceedings of}} \bibinfo{series}{STOC '71}, pp.
  \bibinfo{pages}{151--158}, \doi{10.1145/800157.805047}.

\bibitemdeclare{inproceedings}{EHMMVV13}
\bibitem{EHMMVV13}
\bibinfo{author}{Johannes \surnamestart Ebbing\surnameend},
  \bibinfo{author}{Lauri \surnamestart Hella\surnameend}, \bibinfo{author}{Arne
  \surnamestart Meier\surnameend}, \bibinfo{author}{Julian-Steffen
  \surnamestart M{\"u}ller\surnameend}, \bibinfo{author}{Jonni \surnamestart
  Virtema\surnameend} \& \bibinfo{author}{Heribert \surnamestart
  Vollmer\surnameend} (\bibinfo{year}{2013}): \emph{\bibinfo{title}{Extended
  Modal Dependence Logic}}.
\newblock In: {\sl \bibinfo{booktitle}{WoLLIC}}, pp. \bibinfo{pages}{126--137},
  \doi{10.1007/978-3-642-39992-3\_13}.

\bibitemdeclare{inproceedings}{EbbingL12}
\bibitem{EbbingL12}
\bibinfo{author}{Johannes \surnamestart Ebbing\surnameend} \&
  \bibinfo{author}{Peter \surnamestart Lohmann\surnameend}
  (\bibinfo{year}{2012}): \emph{\bibinfo{title}{Complexity of Model Checking
  for Modal Dependence Logic.}}
\newblock In: {\sl \bibinfo{booktitle}{SOFSEM}}, {\sl \bibinfo{series}{LNCS}} \bibinfo{volume}{7147},
  \bibinfo{publisher}{Springer}, pp. \bibinfo{pages}{226--237},
  \doi{10.1007/978-3-642-27660-6\_19}.

\bibitemdeclare{inproceedings}{ebloya}
\bibitem{ebloya}
\bibinfo{author}{Johannes \surnamestart Ebbing\surnameend},
  \bibinfo{author}{Peter \surnamestart Lohmann\surnameend} \&
  \bibinfo{author}{Fan \surnamestart Yang\surnameend} (\bibinfo{year}{2013}):
  \emph{\bibinfo{title}{Model Checking for Modal Intuitionistic Dependence
  Logic}}.
\newblock In: {\sl \bibinfo{booktitle}{TbiLLC 2011}}, {\sl \bibinfo{series}{LNCS}} \bibinfo{volume}{7758}, \bibinfo{publisher}{Springer}, pp. \bibinfo{pages}{231--256},
  \doi{10.1007/978-3-642-36976-6\_15}.

\bibitemdeclare{article}{galliani13}
\bibitem{galliani13}
\bibinfo{author}{Pietro \surnamestart Galliani\surnameend}
  (\bibinfo{year}{2013}): \emph{\bibinfo{title}{The Dynamification of Modal
  Dependence Logic}}.
\newblock {\sl \bibinfo{journal}{Journal of Logic, Language and Information}}
  \bibinfo{volume}{22}(\bibinfo{number}{3}), pp. \bibinfo{pages}{269--295},
  \doi{10.1007/s10849-013-9175-7}.

\bibitemdeclare{article}{Graedel:2013}
\bibitem{Graedel:2013}
\bibinfo{author}{Erich \surnamestart Gr{\"a}del\surnameend}
  (\bibinfo{year}{2013}): \emph{\bibinfo{title}{Model-Checking Games for Logics
  of Incomplete Information}}.
\newblock {\sl \bibinfo{journal}{Theoretical Computer Science, Special Issue
  dedicated to GandALF 2011}}, pp. \bibinfo{pages}{2--14},
  \doi{10.1016/j.tcs.2012.10.033}.

\bibitemdeclare{inproceedings}{Groenendijk07}
\bibitem{Groenendijk07}
\bibinfo{author}{Jeroen \surnamestart Groenendijk\surnameend}
  (\bibinfo{year}{2007}): \emph{\bibinfo{title}{Inquisitive Semantics: Two
  Possibilities for Disjunction.}}
\newblock In: {\sl \bibinfo{booktitle}{TbiLLC 2007}}, {\sl \bibinfo{series}{LNCS}} \bibinfo{volume}{5422}, \bibinfo{publisher}{Springer}, pp.
  \bibinfo{pages}{80--94}, \doi{10.1007/978-3-642-00665-4\_8}.

\bibitemdeclare{inproceedings}{HeLuSaVi14}
\bibitem{HeLuSaVi14}
\bibinfo{author}{Lauri \surnamestart Hella\surnameend}, \bibinfo{author}{Kerkko
  \surnamestart Luosto\surnameend}, \bibinfo{author}{Katsuhiko \surnamestart
  Sano\surnameend} \& \bibinfo{author}{Jonni \surnamestart Virtema\surnameend}
  (\bibinfo{year}{2014}): \emph{\bibinfo{title}{The Expressive Power of Modal
  Dependence Logic}}.
\newblock In: {\sl \bibinfo{booktitle}{Advances in Modal Logic 2014}}.
\newblock \urlprefix\url{http://arxiv.org/abs/1406.6266}.

\bibitemdeclare{incollection}{hisa89}
\bibitem{hisa89}
\bibinfo{author}{Jaakko \surnamestart Hintikka\surnameend} \&
  \bibinfo{author}{Gabriel \surnamestart Sandu\surnameend}
  (\bibinfo{year}{1989}): \emph{\bibinfo{title}{Informational Independence as a
  Semantical Phenomenon}}.
\newblock In \bibinfo{editor}{J.~E. \surnamestart Fenstad\surnameend},
  \bibinfo{editor}{I.~T. \surnamestart Frolov\surnameend} \&
  \bibinfo{editor}{R.~\surnamestart Hilpinen\surnameend}, editors: {\sl
  \bibinfo{booktitle}{Logic, Methodology and Philosophy of Science}},
  \bibinfo{volume}{8}, \bibinfo{publisher}{Elsevier},
  \bibinfo{address}{Amsterdam}, pp. \bibinfo{pages}{571--589},
  \doi{10.1007/978-94-017-2531-6_3}.

\bibitemdeclare{article}{Hodges97c}
\bibitem{Hodges97c}
\bibinfo{author}{Wilfrid \surnamestart Hodges\surnameend}
  (\bibinfo{year}{1997}): \emph{\bibinfo{title}{Compositional Semantics for a
  Language of Imperfect Information}}.
\newblock {\sl \bibinfo{journal}{Logic Journal of the IGPL}}
  \bibinfo{volume}{5}(\bibinfo{number}{4}), pp. \bibinfo{pages}{539--563},
  \doi{10.1093/jigpal/5.4.539}.

\bibitemdeclare{phdthesis}{jarmo}
\bibitem{jarmo}
\bibinfo{author}{Jarmo~A. \surnamestart Kontinen\surnameend}
  (\bibinfo{year}{2010}): \emph{\bibinfo{title}{Coherence and Complexity in
  Fragments of Dependence Logic}}.
\newblock Ph.D. thesis, \bibinfo{school}{University of Amsterdam}.
\newblock
  \urlprefix\url{http://www.illc.uva.nl/Research/Publications/Dissertations/DS-2010-05.text.pdf}.

\bibitemdeclare{inproceedings}{KoKuLoVi:2011}
\bibitem{KoKuLoVi:2011}
\bibinfo{author}{Juha \surnamestart Kontinen\surnameend},
  \bibinfo{author}{Antti \surnamestart Kuusisto\surnameend},
  \bibinfo{author}{Peter \surnamestart Lohmann\surnameend} \&
  \bibinfo{author}{Jonni \surnamestart Virtema\surnameend}
  (\bibinfo{year}{2011}): \emph{\bibinfo{title}{Complexity of Two-Variable
  Dependence Logic and IF-Logic}}.
\newblock In:\bibinfo{series}{ LICS '11},
  \bibinfo{publisher}{IEEE Computer Society}, pp. \bibinfo{pages}{289--298}, \doi{10.1109/LICS.2011.14}.

\bibitemdeclare{article}{la77}
\bibitem{la77}
\bibinfo{author}{Richard~E. \surnamestart Ladner\surnameend}
  (\bibinfo{year}{1977}): \emph{\bibinfo{title}{The Computational Complexity of
  Provability in Systems of Modal Propositional Logic}}.
\newblock {\sl \bibinfo{journal}{Siam Journal on Computing}}
  \bibinfo{volume}{6}(\bibinfo{number}{3}), pp. \bibinfo{pages}{467--480},
  \doi{10.1137/0206033}.

\bibitemdeclare{article}{Lev1973}
\bibitem{Lev1973}
\bibinfo{author}{Leonid~A. \surnamestart Levin\surnameend}
  (\bibinfo{year}{1973}): \emph{\bibinfo{title}{Universal search problems}}.
\newblock {\sl \bibinfo{journal}{Problems of Information Transmission}}
  \bibinfo{volume}{9}(\bibinfo{number}{3}).

\bibitemdeclare{phdthesis}{lohmannthesis}
\bibitem{lohmannthesis}
\bibinfo{author}{Peter \surnamestart Lohmann\surnameend}
  (\bibinfo{year}{2012}): \emph{\bibinfo{title}{Computational Aspects of
  Dependence Logic}}.
\newblock Ph.D. thesis, \bibinfo{school}{Leibniz Universit\"{a}t Hannover}.
\newblock \urlprefix\url{http://arxiv.org/abs/1206.4564}.

\bibitemdeclare{article}{lohvo13}
\bibitem{lohvo13}
\bibinfo{author}{Peter \surnamestart Lohmann\surnameend} \&
  \bibinfo{author}{Heribert \surnamestart Vollmer\surnameend}
  (\bibinfo{year}{2013}): \emph{\bibinfo{title}{Complexity Results for Modal
  Dependence Logic}}.
\newblock {\sl \bibinfo{journal}{Studia Logica}}
  \bibinfo{volume}{101}(\bibinfo{number}{2}), pp. \bibinfo{pages}{343--366},
  \doi{10.1007/s11225-013-9483-6}.

\bibitemdeclare{inproceedings}{vollmer13}
\bibitem{vollmer13}
\bibinfo{author}{Julian-Steffen \surnamestart M\"{u}ller\surnameend} \&
  \bibinfo{author}{Heribert \surnamestart Vollmer\surnameend}
  (\bibinfo{year}{2013}): \emph{\bibinfo{title}{Model Checking for Modal
  Dependence Logic: An Approach through Post's Lattice}}.
\newblock In: {\sl
  \bibinfo{booktitle}{WoLLIC 2013}}, pp. \bibinfo{pages}{238--250},
  \doi{10.1007/978-3-642-39992-3\_21}.

\bibitemdeclare{article}{Peterson2001}
\bibitem{Peterson2001}
\bibinfo{author}{G.~\surnamestart Peterson\surnameend},
  \bibinfo{author}{J.~\surnamestart Reif\surnameend} \&
  \bibinfo{author}{S.~\surnamestart Azhar\surnameend} (\bibinfo{year}{2001}):
  \emph{\bibinfo{title}{Lower bounds for multiplayer noncooperative games of
  incomplete information}}.
\newblock {\sl \bibinfo{journal}{Comput. \& Math. Appl.}}
  \bibinfo{volume}{41}(\bibinfo{number}{7-8}), pp. \bibinfo{pages}{957 -- 992},
  \doi{10.1016/S0898-1221(00)00333-3}.

\bibitemdeclare{inproceedings}{sano11}
\bibitem{sano11}
\bibinfo{author}{Katsuhiko \surnamestart Sano\surnameend}
  (\bibinfo{year}{2011}): \emph{\bibinfo{title}{First-Order Inquisitive Pair
  Logic}}.
\newblock In \bibinfo{editor}{Mohua \surnamestart Banerjee\surnameend} \&
  \bibinfo{editor}{Anil \surnamestart Seth\surnameend}, editors: {\sl
  \bibinfo{booktitle}{4th Indian Conference, ICLA 2011}}, {\sl
  \bibinfo{series}{LNCS}} \bibinfo{volume}{6521},
  pp. \bibinfo{pages}{147--161}, \doi{10.1007/978-3-642-18026-2\_13}.

\bibitemdeclare{article}{Sevenster:2009}
\bibitem{Sevenster:2009}
\bibinfo{author}{Merlijn \surnamestart Sevenster\surnameend}
  (\bibinfo{year}{2009}): \emph{\bibinfo{title}{Model-theoretic and
  Computational Properties of Modal Dependence Logic}}.
\newblock {\sl \bibinfo{journal}{J. Log. Comput.}}
  \bibinfo{volume}{19}(\bibinfo{number}{6}), pp. \bibinfo{pages}{1157--1173},
  \doi{10.1093/logcom/exn102}.

\bibitemdeclare{inproceedings}{Stockmeyer:1973}
\bibitem{Stockmeyer:1973}
\bibinfo{author}{L.~J. \surnamestart Stockmeyer\surnameend} \&
  \bibinfo{author}{A.~R. \surnamestart Meyer\surnameend}
  (\bibinfo{year}{1973}): \emph{\bibinfo{title}{Word Problems Requiring
  Exponential Time(Preliminary Report)}}.
\newblock In: {\sl \bibinfo{booktitle}{Proceedings of }} \bibinfo{series}{STOC '73},
  \bibinfo{publisher}{ACM}, \bibinfo{address}{New York, NY, USA}, pp.
  \bibinfo{pages}{1--9}, \doi{10.1145/800125.804029}.

\bibitemdeclare{phdthesis}{jonnithesis}
\bibitem{jonnithesis}
\bibinfo{author}{Jonni \surnamestart Virtema\surnameend}
  (\bibinfo{year}{2014}): \emph{\bibinfo{title}{Approaches to Finite Variable
  Dependence: Expressiveness and Computational Complexity}}.
\newblock Ph.D. thesis, \bibinfo{school}{University of Tampere}.
\newblock \urlprefix\url{http://urn.fi/URN:ISBN:978-951-44-9472-7}.

\bibitemdeclare{book}{va07}
\bibitem{va07}
\bibinfo{author}{Jouko \surnamestart V{\"a}{\"a}n{\"a}nen\surnameend}
  (\bibinfo{year}{2007}): \emph{\bibinfo{title}{Dependence Logic - A New
  Approach to Independence Friendly Logic}}.
\newblock {\sl \bibinfo{series}{London Mathematical Society student
  texts}}~\bibinfo{volume}{70}, \bibinfo{publisher}{Cambridge University
  Press},
	\doi{10.1017/CBO9780511611193}.

\bibitemdeclare{incollection}{va08}
\bibitem{va08}
\bibinfo{author}{Jouko \surnamestart V{\"a}{\"a}n{\"a}nen\surnameend}
  (\bibinfo{year}{2008}): \emph{\bibinfo{title}{Modal Dependence Logic}}.
\newblock In \bibinfo{editor}{Krzysztof~R. \surnamestart Apt\surnameend} \&
  \bibinfo{editor}{Robert \surnamestart van Rooij\surnameend}, editors: {\sl
  \bibinfo{booktitle}{New Perspectives on Games and Interaction}}, {\sl
  \bibinfo{series}{Texts in Logic and Games}}~\bibinfo{volume}{4}, pp.
  \bibinfo{pages}{237--254}.
\newblock \urlprefix\url{http://dare.uva.nl/document/130061}.


\bibitemdeclare{phdthesis}{fanthesis}
\bibitem{fanthesis}
\bibinfo{author}{Fan \surnamestart Yang\surnameend} (\bibinfo{year}{2014}):
  \emph{\bibinfo{title}{On Extensions and Variants of Dependence Logic}}.
\newblock Ph.D. thesis, \bibinfo{school}{University of Helsinki}.
\newblock
  \urlprefix\url{http://www.math.helsinki.fi/logic/people/fan.yang/dissertation_fyang.pdf}.

\end{thebibliography}

\end{document}